\documentclass[8pt,twocolumn]{extarticle}

\usepackage[square,numbers]{natbib}

\usepackage{booktabs}
\usepackage{hyperref}

\usepackage[libertine]{newtx}

\usepackage{authblk}

\usepackage{geometry}
\usepackage{caption,float}
\usepackage{titlesec}
\titleformat{\section}
	{\normalfont\large\bfseries}{\thesection}{1em}{}
\titleformat{\subsection}
	{\normalfont\large\bfseries}{\thesubsection}{1em}{}
\titleformat{\subsubsection}[runin]
	{\normalfont\normalsize\itshape}{\thesubsubsection}{1em}{}
\newcommand{\periodafter}[1]{#1.}
\titleformat{\paragraph}[runin]
	{\normalfont\normalsize\itshape}{\theparagraph}{1ex}{\periodafter}

\titlespacing*{\section}{0pt}{2ex}{1ex}
\titlespacing*{\subsection}{0pt}{2ex}{1ex}

\usepackage{fancyhdr}
\fancypagestyle{firstpage}{
  \fancyhf{}
  
  \fancyfoot[L]{\small \textcopyright~ \'{E}ric Goubault, Nuwan Herath Mudiyanselage, Sylvie Putot, 2026. This is the author's version of The Work. It is posted here for your personal use. Not for redistribution. The definitive version was published in \textit{29th ACM International Conference on Hybrid Systems: Computation and Control (HSCC ’26)}, \url{https://doi.org/10.1145/3801146.3805656}.}
}

\usepackage{tikz}

\newcommand{\real}{\mathbb{R}}
\newcommand{\R}{\mathcal{R}}

\newcommand{\oraclein}{\mathcal{O}^{\text{IN}}}
\newcommand{\oracleout}{\mathcal{O}^{\text{OUT}}}

\def\bfm#1{\protect{\makebox{\boldmath $#1$}}}

\DeclareMathAlphabet\mathbfcal{OMS}{cmsy}{b}{n}
\newcommand{\x}{\mathbf{x}}
\newcommand{\X}{\mathbb{X}}
\newcommand{\p}{\mathbf{p}}

\newcommand{\capitalp}{\mathbb{P}}

\newcommand{\z}{\mathbf{z}}
\newcommand{\D}{\mathbb{D}}

\newcommand{\0}{\mathbf{0}}

\newcommand{\List}{\mathcal{L}}
\newcommand{\Sin}{\Sigma^-}
\newcommand{\Sout}{\left(\Sigma^\complement\right)^-}
\newcommand{\Sdelta}{\Sigma^\Delta}
\newcommand{\G}{\mathbb{G}}

\DeclareMathOperator{\width}{width}

\newcommand{\Rin}{{\R}^-}
\newcommand{\Rout}{{\R}^+}
\newcommand{\Xin}{{\mathbb{X}}^-}
\newcommand{\Xout}{{\mathbb{X}}^+}

\usepackage{amsthm}
\theoremstyle{plain}
\newtheorem{thm}{Theorem}[section]
\newtheorem{lem}[thm]{\normalfont\scshape Lemma}

\theoremstyle{definition}
\newtheorem{exmp}{\normalfont\itshape Example}[section]
\theoremstyle{remark}
\newtheorem*{rem}{Remark}

\usepackage{epsfig}

\usepackage{algpseudocode}
\usepackage{algorithm}
\algnewcommand\algorithmicinput{\textbf{Input:}}
\algnewcommand\algorithmicoutput{\textbf{Output:}}
\algnewcommand\Input{\item[\algorithmicinput]}%
\algnewcommand\Output{\item[\algorithmicoutput]}%
\algnewcommand\algorithmicbreak{\textbf{break}}
\algnewcommand\Break{\algorithmicbreak}%

\usepackage{mathtools}

\usepackage{dashbox}

\usepackage{multirow}

\usepackage{diagbox}

\usepackage{pgfplots}
\pgfplotsset{compat=1.18}

\usepgfplotslibrary{fillbetween}

\begin{document}

\title{\Huge\sffamily\bfseries Solving quantified constraint satisfaction problems with applications to control}

\author{\href{https://orcid.org/0000-0002-3198-1863}{\'{E}ric Goubault}}
\author{\href{https://orcid.org/0000-0002-2937-5821}{Nuwan Herath Mudiyanselage}}
\author{\href{https://orcid.org/0000-0001-5624-3755}{Sylvie Putot}}
\affil{\small LIX, CNRS, École polytechnique, Institut Polytechnique de Paris, Palaiseau, France}
\affil{\texttt{\{goubault,herathmudiyanselage,putot\}@lix.polytechnique.fr}}
\date{}

\maketitle
\thispagestyle{firstpage}

\begin{abstract}
Many problems in control can be expressed as quantified constraints, such as robust forward and backward reachability, synthesis of controllers, solutions to motion planning problems. In this paper, we develop a method for computing inner and outer approximations of global solutions to very general quantified constraints, generalizing widely the state of the art on similar problems, usually restricted to at most one alternation of quantifiers or to  local solutions. By global solutions, we mean characterizing accurately the entire set of solutions, within some prescribed bounds. For this, our approach uses a set inversion technique relying on oracles for set inclusion using local solutions to related quantified reachability problems. We exemplify the approach throughout the paper on different examples related to control problems. 
\end{abstract}

\section*{CSS Concepts}

\textbullet~ \textbf{Theory of computation} \textrightarrow~ \textbf{Logic and verification}.

\section*{Keywords}
Set-based methods,  set inversion, quantified constraint satisfaction problem (QCSP), reachability, control

\section{Introduction}

\label{sec:introduction}

The solution set of a variety of problems can be expressed as sets constrained by arbitrary quantified formulas~\cite{ratschan2025applicationsquantifiedconstraintsolving}, or solution of quantified constraint satisfaction problems (QCSP). In the context of the  analysis of control systems, let us mention for example the problem of characterizing the region of the space for which their exists a controller leading to a target region whatever the disturbances in a given bounded set.
\cite{hscc24,nahs25} study the quantified problem of characterizing by inner and outer-approximations the set of states robustly reachable by some function, which can typically be the solution flow of a dynamical system. Each state of the robustly reachable states computable is required to be exactly reachable, whatever some disturbances: each $z$ of the set defined for a given time $t$ in Equation (\ref{ex-hscc}) 
\begin{equation}
\label{ex-hscc}
\{ z \mid \exists x_0 \in \mathbb{X}_0, \ \exists u \in \mathbb{U}, \forall w \in \mathbb{W},  \\ z=\varphi(t;x_0,u,w)\}
\end{equation}
should be exactly reachable whatever the disturbance $w$, which is obviously too strict a requirement in general to be possible. Some relaxations are thus studied instead, such as the time relaxation:
\begin{equation*}
  \{ z  \mid \exists x_0 \in \mathbb{X}_0, \ \exists u \in \mathbb{U},  \ \forall w \in {\mathbb W}, \ \exists s \in [0,T], \ 
   z=\varphi(s;x_0,u,w)\}
\end{equation*}
leading to more general quantified expressions.
In the present work, we rely on the inspiration and techniques of \cite{hscc24,nahs25} but study the inverse problem of a relaxed version. For instance, we want now to characterize in Equation~(\ref{eq:robtime2}), for $\G$ a given target or goal region,  the subset of the initial states $\mathbb{X}_0$ for which region $\G$ is robustly reachable:
\begin{equation}
\label{eq:robtime2}
   \{  x_0 \in \mathbb{X}_0 \mid \ \exists u \in \mathbb{U},  \ \forall w \in {\mathbb W}, \ \exists s \in [0,T], \ 
   \varphi(t;x_0,u,w) \in \G \}
\end{equation}
More generally, we  consider  the inverse problem of characterizing the subset of an input domain, whose image by any function defined as a quantified expression with any possible universal and existential quantifier alternation reaches a goal output. This can be used for various applications in control, (robust) forward and backward reachability, but also existence and synthesis of controllers satisfying some constraints, motion planning problems etc.
Nonetheless, we are not quantifying over functions here, which would be a much more intricate problem to solve. 

\paragraph{Problem statement}

\label{sec:problemstatement}
Let $f$ be a function from $\real^{n+2l}$ to $\real^m$. We want to characterize the set
\begin{equation}\label{eq:generalized_UE_solution_set}
    \Sigma = \left\lbrace \x \in \D \, | \, \forall \p_1 \in \capitalp_1, \dots, \exists \p_{2l} \in \capitalp_{2l}, f(\x,\p) \in \G \right\rbrace
\end{equation}
with the quantifiers $\forall$ and $\exists$ alternating $l$ times, and where $\D$, $\capitalp_1, \dots, \capitalp_{2l}$ and $\G$ are boxes (or intervals in one dimension). Note that this notation includes any alternation ($\forall\exists$, $\forall\exists\forall$, $\exists\forall$, \dots), since the results hold even when $\capitalp_k = \emptyset$.

We note $\capitalp$ the collection of domains $\capitalp_1,\ldots,\capitalp_{2l}$ together with their respective quantifiers $Q_1=\forall,\ Q_2=\exists,\dots,Q_{2l}=\exists$. When needed, interval $\capitalp_i$ (resp. $\X$) together with its quantifier $Q$ is written as $\capitalp_i^Q$ (resp. $\X^Q$). We note $\neg \capitalp$ the collection of domains $\capitalp_1,\ldots,\capitalp_{2l}$ with quantifiers $\neg Q_1=\exists,\ \neg Q_2=\forall, \dots,\neg Q_{2l}=\forall$. 
We note also $\Sigma^\complement$ the complement in $\D$ of $\Sigma$,  defined by: 
\begin{equation}\label{eq:generalized_UE_solution_set_complement}
    \Sigma^\complement = \left\lbrace \x \in \D \, | \, 
		\exists \p_1 \in \capitalp_1, \dots, \forall \p_{2l} \in \capitalp_{2l}, f(\x,\p) \in {\G}^\complement
        \right\rbrace .
\end{equation}

For a set $\mathbb{X}$, we will note $\Xin$ an inner-approximation of $\mathbb{X}$, such that $\forall x \in \Xin, x \in \mathbb{X}$, and $\Xout$ an outer-approximation of  $\mathbb{X}$,  such that $\forall x \in \mathbb{X}, x \in \Xout$.
In this work, relying on local inner and outer-approximations such as in~\cite{hscc24,nahs25}, we  compute global inner and outer-approximations of $\Sigma$.

\begin{exmp}[Running example]
\label{eq:runningexample}
We will exemplify different aspects of our approach on the following quantified problem:
\[
    \Sigma_1 = \left\lbrace x \in \real \, | \, \forall \p_1 \in \left[0,\frac14\right],  \ p^2_1-(x-1)(x-2)(x-3) \in \left[-\frac14,\frac14\right] \right\rbrace
\]
The solution set of this problem can be explained geometrically as follows: this is the set of $x$ such that the points on the cubic curve $(C): \ x\rightarrow -(x-1)(x-2)(x-3)$ have a value between the two horizontal lines $y=-1/4$ and $y=3/16$, as depicted in Figure \ref{fig:runex}. Indeed,  for $x$ to be in $\Sigma_1$, the curve 
$(C)$ at $x$ should take a value in $[-1/4,1/4]-p_1^2$, whatever the value of $p_1^2$ in $[0,1/16]$. This implies that $(C)$ at $x$ should lie in the intersection, for all $p_1^2 \in [0,1/16]$, of sets $[-1/4,1/4]-p_1^2$, which is equal to $[-1/4,3/16]$. 
Geometrically, we see that this will consist in  three disconnected interval solutions around 1, 2 and 3 respectively. Any local method would only be able to find a portion of the solution set within one of these three intervals. In this work, we characterize the full set of solutions, i.e. produce inner and outer approximations of these three intervals. 

\begin{figure}
\begin{center}
\begin{tikzpicture}[yscale=0.95]
  \begin{axis}[
      domain=0.5:3.5,
      samples=400,
      thick,
      grid=major,
      xlabel={$x$},
      ylabel={$y$},
      axis x line=bottom,     
      axis y line=left,
      legend style={at={(0.1,0.98)},anchor=north west,font=\small},
      ymin=-1, ymax=1.5,
    ]

    \addplot[black, very thick, name path=f] {-(x-1)*(x-2)*(x-3)};
    \addlegendentry{$f(x) = - (x - 1)(x - 2)(x - 3)$}

    \addplot[dashed, blue, thick, name path=ymin] {-1/4};
    \addlegendentry{$y = -\tfrac{1}{4}$}
    \addplot[dashed, red, thick, name path=ymax] {3/16};
    \addlegendentry{$y = \tfrac{3}{16}$}

\addplot[purple!70!black, ultra thick, domain=0.5:3.5, samples=200, 
    restrict expr to domain={-(x-1)*(x-2)*(x-3)}{-0.25:0.1875}] 
    {0};
\addlegendentry{$ \{x | - \frac14 \leq f(x) \leq \frac{3}{16} \} $}
  
\addplot[green!70!black, ultra thick, domain=0.5:3.5, samples=200, 
    restrict y to domain=-0.25:0.1875] 
    {-(x-1)*(x-2)*(x-3)};

    \addplot[fill=orange!30, opacity=0.4] fill between[of=ymin and ymax];

  \end{axis}
\end{tikzpicture}
\caption{Geometric illustration of the Example \ref{eq:runningexample}}
\label{fig:runex}
\end{center}
\end{figure}
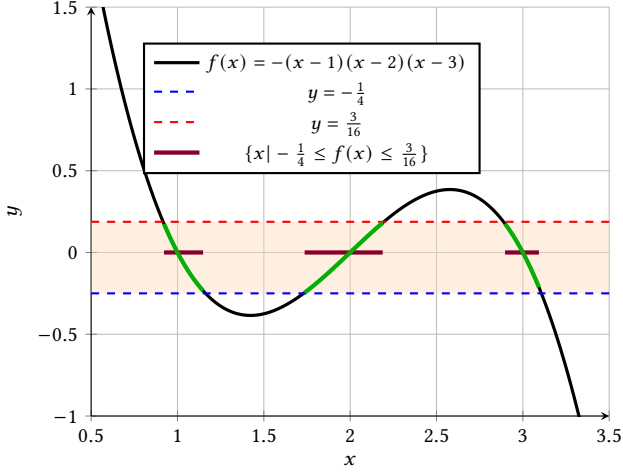

\end{exmp}

\paragraph{Related work}
The applications of quantifier elimination to control design  are numerous \cite{QEcontrol,Jirstrand1997NonlinearCS}. 
Most existing works to solve quantified constraints satisfaction problems rely on quantifier elimination, meaning the symbolic computation of a quantifier-free equivalent form, in particular by means of Cylindrical Algebraic Decomposition~\cite{Collins}. 
 However, the practical applicability of symbolic quantifier elimination to control design and reachability problems is limited in practice. Firstly, because of the high complexity of quantifier elimination. Secondly, reachability is in general not solvable by algebraic methods, although some restricted forms of reachability can be investigated using semi-algebraic tools~\cite{Jirstrand1997NonlinearCS,sturm:hal-03142063,10.5555/646881.710601}. 

Other methods rely on numerical or set-based methods, aiming at computing sound approximations of the solution of quantified reachability problems for the analysis of dynamical systems, under the competing influence of control inputs and disturbances. Our work is part of this direction.  
A well-established framework for reachability analysis
is known as Hamilton-Jacobi techniques, that express reachable sets of a continuous-time dynamical system as the sublevel set of the Hamilton-Jacobi-Isaacs partial differential equation~\cite{10.1109/CDC.2017.8263977}.
These Hamilton-Jacobi methods can produce backward reachable sets~\cite{1463302} with robustness
properties that naturally produce quantified problems~\cite{10.1007/978-3-540-71493-4_34}. For linear time-invariant systems, more scalable set propagation techniques relying on constrained zonotopes
have been proposed to compute inner and outer-approximations of backward reachable sets~\cite{Wetzlinger_2025}. 

Constraint programming approaches have been used to solve some restricted cases of quantified constraint satisfaction problems~\cite{10.1007/3-540-45349-0_7, 10.1145/1183278.1183282}.
Quantified set inversion~\cite{GSJ01,Her06} also considers global inner and outer-approximations of quantified constraint satisfaction problems, but only with one alternation of quantifiers, either of the form $\forall \exists$ (with the so-called $*$-semantic extension) or $\exists \forall$ (with the so-called $**$-semantic extension). This is based on an extension of Kaucher arithmetic \cite{Kaucher} to deal with so-called ``multi-incident'' functions, that is, functions which uses the values of at least one of their input variables multiple times. Also, the work reported in \cite{GSJ01,Her06} is limited to functions which are explicitly given as rational functions (with potential extensions to some transcendental functions such as $sin$, $exp$ etc), which may require extra rewriting to be in ``tree-optimal" form, introducing also copies of the input variables in the multi-incident case. More general approximate quantifier elimination over the theory of Non-Linear Real Arithmetic has been considered in \cite{practicalapproximate}, using a combination of dynamic gridding with the application of Handelman's theorem to obtain a sequence of linear programs. 

Satisfiability Modulo Theory (SMT) also consider quantified logic, in particular 
dReal/dReach \cite{dReal} with support for some quantified SMT
modulo the theory of real numbers and modulo ODEs. Such SMT
solvers do not synthesize the set of states satisfying a quantified formula as we do, but can be used for checking this set is correct. The complexity is
also much higher that what we are proposing, and dReal/dReach is
limited to the exists-forall fragment.

The method we are proposing is not restricted to any number of quantifier alternations nor to explicit functions. We only need to be able to compute outer-approximations of the partial derivatives of these functions, which also enables the application of our method to flow functions given as solutions of differential equations, see e.g. Example \ref{ex:Dubbins}, and hence general reachability problems.

\paragraph{Contributions}

We rely on and extend the techniques of \cite{hscc24,nahs25}, for local inner and outer approximations of quantified equality problems, and use set inversion via interval analysis~\cite{JW93}, to compute global inner and outer-approximations of the inverse relaxed problem. For this, we:
\begin{itemize}
\item propose a partitioning algorithm based on a paving of  the free variables domain  $\D$ to characterize inner and outer approximations of the solution set, relying on the SIVIA set inversion algorithm~\cite{JW93} and generalizing the quantified set inversion of~\cite{HSVJ05}
\item propose oracles to prove whether a current set belongs to the solution set of the quantified problem, to its complement, or has an unknown status, relying on the local forward inner and outer-approximations of \cite{hscc24,nahs25}
\item demonstrate that we can partition also some of the parameter sets in the resolution,  depending on the quantifier alternation, to refine the precision,  thus avoiding  some further partitioning of the free variables domain $\D$.
\end{itemize}

\section{Background on local approximations}
\label{sec:localgr}

Previous work~\cite{lcss20,hscc24,nahs25} have proposed interval-based techniques to compute inner and outer approximations of quantified reachability sets
with $n$ alternations of quantifiers $\forall \exists$ of a form closely related to Equation~(\ref{eq:generalized_UE_solution_set}),
and defined for a function $f: \real^p \to \real^m$ and a box $\X = \X_1 \times \dots \times \X_{2n}$ by:
\begin{multline*}
R_{\bfm p}(f,\X)=\left\{ z \in \real^m \mid \forall \bfm{x}_1 \in \X_1, \ \exists \bfm{x}_2\in \X_2, \ \ldots,\right.\\
\left . \forall \bfm{x}_{2n-1} \in \X_{2n-1}, \exists \bfm{x_{2n}}
\in \X_{2n}, \ z=f(\bfm{x}_1,\bfm{x}_2,\ldots,\bfm{x}_{{2n}})\right\}
\end{multline*}
where $\bfm p$ is the sequence $(j_1, \dots, j_{2n})$ with $j_i= \dim \X_i$ and  $p = \sum_{i=1}^{2n} j_i$.

Relying on a generalized mean value theorem, Theorem~\ref{thm:robustAE} describes the computation of interval inner and outer-approximations in the simple case of one alternation of quantifiers $\forall \exists$
\begin{equation}
\label{eq:RAE}
R_{\forall\exists}(f,\X)=\left\{ z \in \real \mid \forall \bfm{x}_1 \in \X_1, \ \exists \bfm{x}_2\in \X_2, \ z=f(\bfm{x}_1,\bfm{x}_2)\right\}.
\end{equation}
This result will be used throughout the paper for the computations on the running example. 

For a continuously differentiable vector-valued function $f: \real^m \rightarrow \real^n$, we note $f_i$ its $i$-th component and $\nabla f = (\nabla_j f_i)_{ij} = (\frac{\partial f_i}{\partial x_j})_{1\leq i \leq n, 1 \leq j \leq m}$ its Jacobian matrix.
Interval quantities or boxes, whether scalar or vector-valued, will be noted using the set notation, e.g $\X=[\underline{x},\overline{x}]$. 
For a box $\X$, we note $c({\X}) = (\underline{x}+\overline{x})/2$ its  center  and  $r({\X}) = (\overline{x}-\underline{x})/2$ its radius.
\begin{thm}[Thm. 2, \cite{lcss20}]
\label{thm:robustAE}
Let $f$ be a continuously differentiable function from $\real^m$ to $\real$ and $\X = \X_1 \times \X_2$ a box. 
Let $\bfm{F^0}$, $\bfm{\nabla}_{1}$ and  $\bfm{\nabla}_{2}$ be vectors of intervals such that $ f(c(\X)) \subseteq \bfm{F^0}  $, $\{ \left| \nabla_1 f (x_1,x_2) \right|, \, x_1 \in \X_1, \, x_2 \in \X_2  \} \subseteq  \bfm{\nabla}_{1}$ and $\{ \left| \nabla_2 f (x_1,x_2) \right|, x_1 \in \X_1 , \, x_2 \in \X_2 \} \subseteq  \bfm{\nabla}_{2}$.
We have:
\begin{itemize}
     \item  an outer-approximation of $R_{\forall\exists}(f,\X)$:
  \[  [\underline{f^0} - \langle \overline{\nabla}_2 ,  r(\X_2) \rangle + \langle \underline{\nabla}_1 ,  r(\X_1) \rangle ,  \overline{f^0} + \langle \overline{\nabla}_2 , r(\X_2) \rangle - \langle \underline{\nabla}_1 , r(\X_1) \rangle ] 
  \]
  \item an inner-approximation of $R_{\forall\exists}(f,\X)$ when the interval is non-empty: 
\[ [ \overline{f^0}  - \langle \underline{\nabla}_2 , r({\X_2}) \rangle + \langle \overline{\nabla}_1 , r(\X_1) \rangle ,  \underline{f^0} +  \langle \underline{\nabla}_2 , r(\X_2) \rangle - \langle \overline{\nabla}_1 ,  r(\X_1) \rangle ]  \]  
\end{itemize}
\end{thm}
Intuitively, the term corresponding to the existentially quantified variable $\x_2$ in Equation~(\ref{eq:RAE}) contributes to more freedom in its choice thus increasing the width of the reachable sets, 
while  the universally quantified $\x_1$ contributes to decreasing it. 
We refer to \cite{hscc24,nahs25} for solutions in the general case, which are necessary to implement the oracles of Section~\ref{sec:oracles}. 

\section{Global inner and outer-approximations of solutions to QCSP}

Local solutions as in Section \ref{sec:localgr}, are solutions within a neighborhood of some point, generally relying on some approximation only locally valid or accurate. 
We are now looking for global solutions, i.e. for the set of solutions in a large domain.  

\subsection{Paving-based algorithm}

Solving the global reachability problem of Section~\ref{sec:problemstatement} relying on the local reachability approaches introduced in Section~\ref{sec:localgr} can be done by  paving the free domain $\D$.
The local approach will be used on each subset of the current paving, to determine whether it is included in the solution set $\Sigma$, in its complement $\Sigma^\complement$ or has an unknown status, which will require refining the current subset. 

We  thus rely on a classical  paving algorithm \`a la SIVIA \cite{JW93} in Algorithm~\ref{alg:paver}, building on oracles $\oraclein(\X)$ (respectively $\oracleout(\X)$) to decide  whether $\X\subseteq \Sigma$ (respectively $\X \subseteq \Sigma^\complement$ or, equivalently, $\X \cap \Sigma = \emptyset$). These oracles will be described in Section~\ref{sec:oracles}.
In Section~\ref{sec:refinement} we will show that we can refine the algorithm and improve the quality of the oracles, by partitioning of the parameter space $\capitalp$, depending on the quantifiers. 

Note that the two partitioning algorithms are essentially orthogonal: the first one is there to ensure finding a global solution within the free domain $\D$ by collecting sufficiently many (box-like) subsets $\X$ of $\D$ such that $\X \subseteq \Sigma$. Whereas the second one is there to improve the quality of the oracles.

\begin{algorithm}
	\caption{Paving of $\Sigma$.}
	\label{alg:paver}
	\begin{algorithmic}[1]
		\Input A box $\X_0$ and the precision $\varepsilon$
		\Output A set of boxes $\Sin$ which are in $\Sigma$, a set of boxes $\Sout$ which are outside $\Sigma$, and a set of undefined boxes $\Sdelta$
		\State $\List \gets \{\X_0\}$
		\While {$\List$ not empty}
		    \State $\X \gets$ pop a box from $\List$
			\If {$\oraclein(\X)$}
				\State add $\X$ to $\Sin$
			\ElsIf {$\oracleout(\X)$}
				\State add $\X$ to $\Sout$
            \ElsIf {$\width(\X) \leq \varepsilon$}
				\State add $\X$ to $\Sdelta$
			\Else
				\State bisect $\X$ into $\X_1$ and $\X_2$
				\State add $\X_1$ and $\X_2$ into $\List$
			\EndIf
		\EndWhile
		\State \Return $\Sin$, $\Sout$, $\Sdelta$
	\end{algorithmic}
\end{algorithm}

\subsection{Oracles}
\label{sec:oracles}

We define
$\oraclein(\X)$ as true when
$ \X \subseteq \Sigma, $
which is equivalent to
\[
    \forall \x \in \X, \forall \p_1 \in \capitalp_1, \dots, \exists \p_{2l} \in \capitalp_{2l}, f(\x,\p) \in \G.
\]
\noindent and 
$\oracleout(\X)$ as true when $\X \subseteq \Sigma^\complement$ or
$\X \cap \Sigma = \emptyset$, 
which is equivalent to
\[
	\forall \x \in \X, \exists \p_1 \in \capitalp_1, \dots, \forall \p_{2l} \in \capitalp_{2l}, f(\x,\p) \in \G^\complement.
\]

We introduce the following parameterized set, associated with $\Sigma$ for a free domain $\X$,  quantified parameter domain $\capitalp$ as in the Problem Statement, Section \ref{sec:introduction}, 
a goal domain $\G$ and a function $f$: 
\begin{equation}\label{eq:R}
	\begin{split}
		\R(\X^Q, \capitalp, \G) = \lbrace \zeta \in \real^m \, | \, Q \x \in \X, \forall \p_1 \in \capitalp_1, \dots, \exists  \p_{2l} \in \capitalp_{2l},\\
		\exists \z \in \G, f(\x,\p) - \z = \zeta \rbrace
	\end{split}
\end{equation}

We will be using also sets $\R$ as above, but applied to $\X^Q$,  $\neg \capitalp$ (see again the Problem Statement, Section \ref{sec:introduction}) and $\G^\complement$, that we spell out below for the sake of clarity:
\begin{equation}\label{eq:Rop}
	\begin{split}
		\R(\X^Q, \neg \capitalp, \G^\complement) = \lbrace \zeta \in \real^m \, | \, Q \x \in \X, \exists \p_1 \in \capitalp_1, \dots, \forall  \p_{2l} \in \capitalp_{2l},\\
		\exists \z \in \G^\complement, f(\x,\p) - \z = \zeta \rbrace
	\end{split}
\end{equation}

Inner and outer-approximations $\R^-(\X^Q, \capitalp, \G)$ and $\R^+(\X^Q, \capitalp, \G)$ of the set $\R(\X^Q, \capitalp, \G)$ can be computed with the approach introduced in Section~\ref{sec:localgr}. 
They will be used to define two possible criteria for each of $\oraclein$ and $\oracleout$. 
These are given in Lemmas \ref{lem:infirst},  \ref{lem:outfirst}, \ref{lem:insecond} and \ref{lem:outsecond}, and summarized in Table~\ref{tab:criteria}.

\begin{table}
	\caption{Summary of the criteria.}
	\label{tab:criteria}
	\begin{tabular}{ccc}
		\toprule
		& Using $\capitalp$ and $\G$ & Using $\neg\capitalp$ and $\G^\complement$ \\
		\midrule
		$\oraclein$ & $\{\0\} \subseteq \Rin(\X^\forall,\capitalp,\G)$ & $\{\0\} \cap \Rout(\X^\exists,\neg\capitalp,\G^\complement) = \emptyset$ \\
		$\oracleout$ & $\{\0\} \cap \Rout(\X^\exists,\capitalp,\G) = \emptyset$ & $\{\0\} \subseteq \Rin(\X^\forall,\neg\capitalp,\G^\complement)$ \\
		\bottomrule
	\end{tabular}
\end{table}

\begin{lem}[$\oraclein$, using $\capitalp$ and $\G$]
\label{lem:infirst}
\[
	\{\0\} \subseteq \Rin(\X^\forall,\capitalp,\G) \implies \oraclein(\X).
\]
\end{lem}

\begin{proof}
	Let us assume that $\{\0\} \subseteq \Rin(X^\forall, \capitalp, \G)$. Since by definition $\Rin(\X^\forall, \capitalp, \G) \subseteq \R(\X^\forall, \capitalp, \G)$, we have $\{\0\} \subseteq \R(\X^\forall, \capitalp, \G)$. In other words, according to Equation~(\ref{eq:R}),
	\[
		\forall \x \in \X, \forall \p_1 \in \capitalp_1, \dots, \exists \p_{2l} \in \capitalp_{2l}, \exists \z \in \G, f(\x,\p) - \z = \0.
	\]
	Subsequently,
	\[
		\forall \x \in \X, \forall \p_1 \in \capitalp_1, \dots, \exists \p_{2l} \in \capitalp_{2l}, f(\x,\p) \in \G
	\]
	According to Equation~(\ref{eq:generalized_UE_solution_set}), that means $X \subseteq \Sigma$.
\end{proof}

\begin{lem}[$\oracleout$, using $\capitalp$ and $\G$]
\label{lem:outfirst}
\[
	\{\0\} \cap \Rout(\X^\exists, \capitalp, \G) = \emptyset \implies \oracleout(\X).
\]
\end{lem}

\begin{proof}
	Let us assume that $\{\0\} \cap \Rout(\X^\exists, \capitalp, \G) = \emptyset$. 
    
    As $\R(\X^\exists, \capitalp, \G) \subseteq \Rout(\X^\exists, \capitalp, \G)$, it implies that $\{\0\} \cap \R(\X^\exists, \capitalp, \G) = \emptyset$. In other words, according to Equation~(\ref{eq:R}),
	\[
		\exists \x \in \X, \forall \p_1 \in \capitalp_1, \dots, \exists \p_{2l} \in \capitalp_{2l}, \exists \z \in \G, f(\x,\p) - \z = \0
	\]
    \noindent i.e.
    \[
		\exists \x \in \X, \forall \p_1 \in \capitalp_1, \dots, \exists \p_{2l} \in \capitalp_{2l}, f(\x,\p) \in \G
	\]
	is not satisfied, hence $X \cap \Sigma =\emptyset$.
\end{proof}

Dually, we can use $\neg \capitalp$ and $\G^\complement$ to decide $\oraclein$ and $\oracleout$. The proofs are dual to the ones above and we do not repeat them. 
\begin{lem}[$\oraclein$, using $\neg \capitalp$ and $\G^\complement$]
\label{lem:insecond}
\[
	\{\0\} \cap \Rout(\X^\exists, \neg \capitalp, \G^\complement)=\emptyset \implies \oraclein(\X).
\]
\end{lem}

\begin{lem}[$\oracleout$, using $\neg \capitalp$ and $\G^\complement$]
\label{lem:outsecond}
\[
	\{\0\} \subseteq \Rin(\X^\forall, \neg \capitalp,\G^\complement) \implies \oracleout(\X).
\]
\end{lem}

\subsection{Sample calculations and discussion}

In this section, we exercise the different criteria given by Lemma \ref{lem:infirst}, \ref{lem:outfirst}, \ref{lem:insecond} and \ref{lem:outsecond} on the running Example \ref{eq:runningexample}. 

We demonstrate how they can be computed using simple 0th order inner and outer approximation methods of Theorem \ref{thm:robustAE}. 
In that particular case, the criteria based on $\capitalp$, $\G$ and on $\neg \capitalp$, $\G^\complement$ agree. This is exemplified in Section \ref{sec:sampleoraclein} with $\oraclein$ for interval $[17/18,19/18]$, which is proven to be in the solution set $\Sigma_1$ by both criteria and in Section \ref{sec:sampleoracleout} for interval $[0,1/2]$ which is proven to be outside the solution set $\Sigma_1$. 
The intervals $[17/18,19/18]$, $[0,1/2]$ and their positioning with respect to the first component of the solution set are illustrated on Figure~\ref{fig:runex2},  a zoomed view of Figure~\ref{fig:runex}.

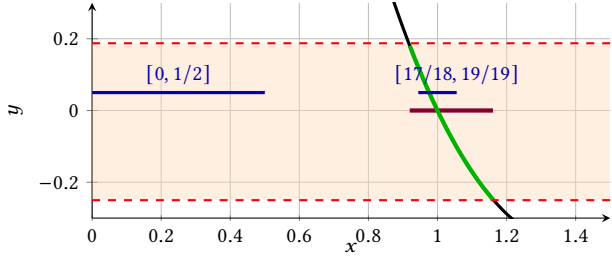
\begin{figure}
\begin{tikzpicture}[yscale=0.5]
  \begin{axis}[
      domain=0.:1.5,
      samples=400,
      thick,
      grid=major,
      xlabel={$x$},
      ylabel={$y$},
      axis x line=bottom,     
      axis y line=left,
      legend style={at={(0.1,0.98)},anchor=north west,font=\small},
      ymin=-0.3, ymax=0.3,
    ]

    \addplot[black, very thick, name path=f] {-(x-1)*(x-2)*(x-3)};

    \addplot[dashed, red, thick, name path=ymin] {-1/4};
    \addplot[dashed, red, thick, name path=ymax] {3/16};

\addplot[purple!70!black, ultra thick, domain=0.:1.5, samples=200, 
    restrict expr to domain={-(x-1)*(x-2)*(x-3)}{-0.25:0.1875}] 
    {0};
  
\addplot[green!70!black, ultra thick, domain=0.:1.5, samples=200, 
    restrict y to domain=-0.25:0.1875] 
    {-(x-1)*(x-2)*(x-3)};
  
    \addplot[fill=orange!30, opacity=0.4] fill between[of=ymin and ymax];

\draw[very thick,blue!70!black,-] (0.9444,0.05) -- (1.055555,0.05) node[above]{$[17/18,19/19]$};


\draw[very thick,blue!70!black,-] (0.,0.05) -- (0.5,0.05) node[above,midway]{$[0,1/2]$};

  \end{axis}
\end{tikzpicture}
\caption{Sample oracle computations on Example \ref{eq:runningexample}}
\label{fig:runex2}
\end{figure}

\subsubsection{Oracle $\oraclein$}

\label{sec:sampleoraclein}

\begin{exmp}
\label{ex:runningexamplebis}
We apply the criterion of Lemma~\ref{lem:infirst} on Example \ref{eq:runningexample}, for $\X=[17/18,19/18]$. We want an inner-approximation of
\begin{multline*}
		\R(\X^\forall, [0,1/4]^\forall, [-1/4,1/4]) = \lbrace \zeta \in \real \, | \, \forall x \in \X, \forall p_1 \in [0,1/4], \\
		\exists z \in [-1/4,1/4], f(x,p_1) - z = \zeta \rbrace
	\end{multline*}
    \noindent with $f(x,p_1)=p^2_1-(x-1)(x-2)(x-3)$.
Let us apply Theorem~\ref{thm:robustAE} on $g$ defined by $g(x,p_1,z)=f(x,p_1)-z$. We first compute $|\nabla g|$ and bound it on $([17/18,19/18],[0,1/4],[-1/4,1/4])$:
\begin{eqnarray*}
|\frac{\partial g}{\partial x}| &=& |-3(x-2)^2 +1| \subseteq [181/108,253/108]\\
|\frac{\partial g}{\partial p_1}|&=&|2p_1| \subseteq [0,1/2]  \\
|\frac{\partial g}{\partial z}|&=&|-1| \subseteq [1,1]
\end{eqnarray*}
For the central value $(x_0,(p_1)_0,z_0)=(1,1/8,0)$, the value of $g$ is $1/64$. 
Thus, an inner approximation $\R^-(\X^\forall, [0,1/4]^\forall, [-1/4,1/4])$ is given by: 
\begin{multline*}
1/64+1/18[253/108,-253/108]+1/8[1/2,-1/2]+1/4[-1,1]\\=[-649/15552,1135/15552]
\end{multline*}
which contains $0$, hence $[17/18,19/18] \subseteq \Sigma_1$. In this computation, the first two intervals, corresponding to the dependency on the universal quantifiers, are adversarial, with lower bound greater than the upper bound, and can be interpreted as Minkowski difference.

Note  that we are unable to find directly a local solution to Example \ref{eq:runningexample} using directly Theorem \ref{thm:robustAE} or more generally the approach of ~\cite{hscc24,nahs25}. Indeed, variable $x$ is defined implicitly in this example and not explicitly, as in  Equation \ref{eq:RAE}. 

Finally, the method from \cite{GSJ01,Her06} using the *-semantic extension of function $g(x,p_1,z)=f(x,p_1)-z=p_1^2-(x-1)(x-2)(x-3)-z$, rewritten in true-optimal, uni-incident form as required by this method, $p_1^2-x_1^3+6x_2^2-11x_3+6-z$ gives: 
\begin{multline*}
[0,1/4]^2-[17/18,19/18]^3+6[19/18,17/18]^2\\-11[17/18,19/18]+6-[1/4,-1/4] 
 = [2.5002,-0.06696]
\end{multline*}
\noindent by noticing that $g$ is decreasing with respect to $x$, $x_1$ and $x_3$, and  increasing in $x_2$. This interval is improper and contains 0, hence $[17/18,19/18]$ is proven to be included in $\Sigma_1$ as well by this method.  
\end{exmp}

\begin{exmp}
\label{ex:runningoutfirst}
We now apply the criterion of Lemma \ref{lem:insecond}. 
to the running example, again on $\X=[17/18,19/18]$. 
We have to compute an outer-approximation of 
\begin{multline*}
		\R(\X^\exists, [0,1/4]^\exists, [-1/4,1/4]^\complement) = \lbrace \zeta \in \real \, | \, \exists x \in \X, \exists p_1 \in [0,1/4], \\
		\exists z \in [-1/4,1/4]^\complement, f(x,p) - z = \zeta \rbrace
	\end{multline*}
    \noindent with $f(x,p)=p^2_1-(x-1)(x-2)(x-3)$.
Using the computation of the partial derivatives of $g$ done in Example \ref{ex:runningexamplebis}, 
we get this time a sum of classical intervals (with a disjunction to represent the complement of $[-1/4,1,4]$), as all quantifiers are existential:
\begin{multline*}
1/64+1/18[-253/108,253/108]+1/8[-1/2,1/2]+\\
(-\infty,-1/4)\cup (1/4,\infty)=
(-\infty,-649/15552) \cup (1135/15552,\infty)
\end{multline*}

\noindent Therefore $\{0\}$ does not intersect $\R^+(\X^\exists,\neg \capitalp,\G^\complement)$ and $[17/18,19/18]$ is also proven to be in $\Sigma_1$ by the criterion of Lemma \ref{lem:insecond}.

We notice in this example that the criterion of Lemma \ref{lem:insecond}, with the inner and outer-approximation method of Theorem \ref{thm:robustAE} agrees with the one of Lemma \ref{lem:infirst} in that case. Indeed, it computes the complement set of the set $\R^-$ computed with the criterion of Lemma \ref{lem:infirst}. This is discussed in Section \ref{sec:comparison}.  

\end{exmp}

\subsubsection{Oracle $\oracleout$}

\label{sec:sampleoracleout}

\begin{exmp}
\label{ex:ex38}
\label{ex:runningoutsecond}
We now apply the criterion of Lemma \ref{lem:outsecond}  with $\X=[0,1/2]$. We need to compute an inner-approximation of:
\begin{multline*}
\R(\X^\forall,[0,1/4]^\exists,[-1/4,1/4]^\complement)=\lbrace \zeta \in \real \, | \, \forall x \in \X, \exists p_1 \in [0,1/4], \\
		\exists z \in [-1/4,1/4]^\complement, f(x,p) - z = \zeta \rbrace
	\end{multline*}
   \noindent with $\X=[0,1/2]$.
The central value of $g$ for this $\X$ is $g(1/4,1/8,0)=29/8$ and the bounds for the partial derivatives of $g$ need to be evaluated on this new domain, in particular $|\frac{\partial g}{\partial x}| \subseteq [23/4,11]$.
Using Theorem~\ref{thm:robustAE}, we get
\begin{multline*}
   \R^-(\X^\forall,[0,1/4]^\exists,[0,1/4]^\complement)=29/8+1/4[11,-11]+1/8[0,0]\\+(-\infty,-1/4)\cup(1/4,\infty)
   =(-\infty, 5/8)\cup (53/8,\infty)
   \end{multline*}
By Lemma \ref{lem:outsecond} we conclude that $\X =[0,1/2]$ is included in $\Sigma_1^\complement$.
\end{exmp}

\subsubsection{Comparison of the various criteria}
\label{sec:comparison}
As the running example illustrates, for the particular case of  dimension 1 with the order 0 method of Theorem \ref{thm:robustAE} (and its extension in \cite{hscc24,nahs25}), 
the dual criteria, Lemma \ref{lem:insecond} and Lemma \ref{lem:outsecond} of Lemma \ref{lem:infirst} and Lemma \ref{lem:outfirst} respectively, are exactly equivalent. 

In higher-dimension, for $m > 1$, and the order 0 method of Theorem \ref{thm:robustAE}, criteria of Lemma \ref{lem:outsecond} and Lemma \ref{lem:infirst} are in general disadvantaged over the ones given in Lemmas \ref{lem:outfirst} and \ref{lem:insecond}. This is because they both use inner-approximations that are abstracted in a  coarser way than outer-approximations, due to the quantifier relaxation method of \cite{hscc24} that requires to further inner-approximate, with different quantifiers, the different components of $\z$. This relaxation is no longer required with higher-order methods such as the ones of \cite{nahs25}, and in particular with the order 1 method using zonotopes. But because higher-order methods abstractions may lose precision differently for outer-approximations and inner-approximations, there is a priori no reason that the dual criteria should always agree, although inner-approximations are usually more difficult to compute than over-approximations. 

\subsubsection{An application to backward reachability}

\label{ex:Dubbins}

As mentioned in the introduction, the approach developed here is well suited not only for computing forward quantified reachable sets, as in \cite{hscc24}, but also backward quantified reachable sets. Moreover, and this will be illustrated on an example in this section, it can  tackle cases in which the function $f$ of the problem defined in Equation (\ref{eq:generalized_UE_solution_set}) is not polynomial, and is not even given analytically at all. 

We consider Example 1 of \cite{hscc24}, which is a simplification of the model from \cite{lcss20}, modeling a Dubins vehicle with uncertainties on the $x$ axis: $\dot{x}=v cos(\theta)+b_1$, $\dot{y}=v sin(\theta)$, $\dot{\theta}=a$. 
We suppose that the speed $v$ is equal to 1 and we have a control period of $t=0.5$. 
The  initial conditions are uncertain given in $\mathbb{X}_0=\{(x,y,\theta) \ \mid \ x \in [-0.1,0.1], \ y \in [-0.1,0.1], \ \theta \in [-0.01,0.01]\}$, the control $a$ can take values in $\mathbb{U}=[-0.01,0.01]$ and disturbance $b_1$ can take values in $\mathbb{W}=[-0.01,0.01]$. Both control and disturbance are supposed to be constant over the control period $[0,0.5]$. 

The quantified constraint satisfiability problem we consider is:
\begin{multline*} 
\Sigma=\left\lbrace x_0 \in [-0.1,0.1] \ \mid \ \exists a\in [-0.01,0.01], \ \exists y_0 \in [-0.1,0.1], \right. \\ \exists \theta_0\in [-0.01,0.01], \ \forall b_1 \in [-0.01,0.01], \ \exists t \in [0,0.5],\\ \left. \varphi_x(t;x_0,y_0,\theta_0,a,b_1)\in \G \right\rbrace
\end{multline*}
\noindent where $\G=[0.49,1]$ 
and $\varphi_x$ is the $x$ component of the flow function $\varphi$. The set $\Sigma$ is the set of initial conditions on  $x$ such that 
the Dubins vehicle, with adversarial perturbation $b_1$, will eventually reach an $x$ coordinate between $0.49$ and $1$ in a time between $0$ and $0.5$. 

Using simple outer-approximation methods~\cite{hscc24,nahs25}, we easily get estimates for a ``central trajectory" of the flow of the differential system, i.e. of the trajectory starting at $x=0$, $y=0$, $\theta=0$, 
$b_1=0$ and $a=0$. The central trajectory plays the role of $f^0$ in Theorem \ref{thm:robustAE}. 

We also recap below the enclosures of the Jacobian $J=(J_{x,q})_q$ of the $x$ component $\varphi_x$ of the flow function $\varphi$ with respect to all parameters and initial conditions $q$, more details can be found in \cite{nahs25}. 
For brevity's sake, instead of giving the Jacobian estimates, we directly provide the lower bounds $I_{x,q}$ for $\langle \underline{J}_{x,q} ,  r(\bf q) \rangle$ (which plays the role of 
$\langle \underline{\nabla} ,  r(\X) \rangle$ in Theorem \ref{thm:robustAE}) and upper bounds $O_{x,q}$ for $\langle \overline{J}_{x,q} ,  r(\bf q) \rangle$ (which plays the role of 
$\langle \overline{\nabla} ,  r(\X) \rangle$ in Theorem \ref{thm:robustAE}). These coefficients are instrumental in the computation of the $\R^-$ and $\R^+$ sets of the criteria of Lemmas \ref{lem:infirst}, \ref{lem:insecond}, \ref{lem:outfirst} and \ref{lem:outsecond}, using the 0th order method of Theorem 1 of \cite{hscc24}, which extends Theorem \ref{thm:robustAE} of Section \ref{sec:localgr} for general quantified problems:

\begin{itemize}
    \item the central trajectory is given by $x_c=t$, $y_c=0$ and $\theta_c=0$, therefore, it is going to be $0$ at time 0
    \item $I_{x,t}=[0,0.495]$, $O_{x,t}=[0,0.505]$, 
    $I_{x,a}=0$, \\$O_{x,a}=[-6.545 \ 10^{-7},6.545 \ 10^{-7}]$, 
    $I_{x,y_0}=O_{x,y_0}=0$, \\
$I_{x,\theta_0}=0$, $O_{x,\theta_0}=[-1.309 \ 10^{-6}, 1.309 \ 10^{-6}]$, \\
$I_{x,b_1}=0$, $O_{x,b_1}=[-0.005,0.0$ $05]$ 
\end{itemize} 

We now show that $[0.05,0.1] \in \Sigma$ and $[-0.1,0] \in \Sigma^\complement$ using, respectively, criteria of Lemma \ref{lem:infirst} and of Lemma \ref{lem:outfirst}. Given that $\capitalp$ is $\exists a\in [-0.01,0.01], \ \exists y_0 \in [-0.1,0.1], \ \exists \theta_0\in [-0.01,0.01], \ \forall b_1 \in [-0.01,0.01], \ \exists t \in [0,0.5]$, and $\G=[0.49,1]$, we compute: 
\begin{multline*}
\R^-([0.05,0.1]^\forall,\capitalp,\G)=0+[0.1,0.05]+0+0\\
+0+[0.005,-0.005]+0+[-1,-0.49]=[-0.895,0.005]
\end{multline*}
\noindent hence the condition of criterion of Lemma \ref{lem:infirst} is satisfied, and $[0.005,0.1]\in \Sigma$. Now, we compute, 
\begin{multline*}
\R^+([-0.1,0]^\exists,\capitalp,\G)=0+[-0.1,0]+0+0+0+\\
[-0.005,0.005]+[0.495,0]+[0.49,1] =[0.88,1.005]
\end{multline*}
\noindent which does not intersect $\{0\}$, therefore, by the criterion of Lemma \ref{lem:outfirst}, $[-0.1,0]\in \Sigma^\complement$.

\section{Improvements of the paving algorithm}

In this section, we show that we can improve the precision of the oracles by partitioning also the parameter space $\capitalp$. 
\label{sec:refinement}

We begin by exposing some basic logical facts in Lemma \ref{lem:weaken}. We will then rely on these facts  in Section~\ref{sec:subdiv1} and Section~\ref{sec:subdiv2} to define sound parameter subdivision based refinements of the approximations of sets $\R^-$ and $\R^+$. 

\begin{lem}
\label{lem:weaken}
Let $\Phi$ be a quantified formula in the first-order theory of real numbers, $\capitalp_1, \dots, \capitalp_{j}$ be subsets of real numbers, and $x_1,\dots,x_j$ be variables for $j\geq 1$. Then for all sets $\capitalp_j^1,\dots,\capitalp_j^{k}$ with $\capitalp_j\supseteq \bigcup\limits_{i=1}^k \capitalp_j^i$: 
\begin{enumerate}
\item $\left(\exists x \in \capitalp_j, \ \Phi\right)\Longleftarrow \bigvee\limits_{i=1}^k \left(\exists x \in \capitalp_j^i, \ \Phi\right)$, and is an equivalence if $\capitalp_j= \bigcup\limits_{i=1}^k \capitalp_j^i$
\item $\left(\forall x \in \capitalp_j, \ \Phi \right)\Longrightarrow \bigwedge\limits_{i=1}^k \left(\forall x \in \capitalp_j^i, \ \Phi\right)$, and is an equivalence if $\capitalp_j= \bigcup\limits_{i=1}^k \capitalp_j^i$
\item We always have: 
\begin{multline*}
    \left(Q_1 x_1 \in \capitalp_1, \dots, Q_{j-1} x_{j-1} \in \capitalp_{j-1}, \exists x_j \in \capitalp_j, \ \Phi\right) \\ \Longleftarrow \bigvee\limits_{i=1}^k \left(Q_1 x_1 \in \capitalp_1, \dots, Q_{j-1} x_{j-1} \in \capitalp_{j-1}, \exists x_j \in \capitalp_j^i, \ \Phi\right)
    \end{multline*}
    \noindent and this is not an equivalence when at least one of the quantifiers $Q_1,\ldots,Q_{j-1}$ is $\forall$, even when $\capitalp_j= \bigcup\limits_{i=1}^k \capitalp_j^i$
\item We always have: 
\begin{multline*}
    \left(Q_1 x_1 \in \capitalp_1, \dots, Q_{j-1} x_{j-1} \in \capitalp_{j-1}, \forall x_j \in \capitalp_j, \ \Phi\right) \\ \Longrightarrow \bigwedge\limits_{i=1}^k \left(Q_1 x_1 \in \capitalp_1, \dots, Q_{j-1} x_{j-1} \in \capitalp_{j-1}, \forall x_j \in \capitalp_j^i, \ \Phi\right)
    \end{multline*}
    \noindent but this is not an equivalence when at least one of the quantifiers $Q_1,\ldots,Q_{j-1}$ is $\exists$, even when $\capitalp_j= \bigcup\limits_{i=1}^k \capitalp_j^i$.
\end{enumerate}
\end{lem}

\begin{proof}
The first two equivalences are obvious. Consider case (3), case (4) being implied by considering the De Morgan dual of formulas of case (3). 
Suppose there exists $i$ within $1,\dots,k$ with $Q_1 x_1 \in \capitalp_1, \dots, Q_{j-1} x_{j-1} \in \capitalp_{j-1}, \exists x_j \in \capitalp_j^i$. Then indeed, as $\capitalp_j^i \subseteq \capitalp_j$,  $Q_1 x_1 \in \capitalp_1, \dots, Q_{j-1} x_{j-1} \in \capitalp_{j-1}, \exists x_j \in \capitalp_j$ holds as well. 

Cases (3) and (4) are in general not equivalences. For instance, for case (3), consider the formula: 
\begin{equation}
\label{eq:counter3}
\forall x \in [-1,1], \ \exists y \in [-1,1], \ x+y=0
\end{equation}
Equation (\ref{eq:counter3}) is obviously true: for all $x \in [-1,1]$, set $y=-x$, which belongs to $[-1,1]$ as well. 
Now, consider the subdivision of $[-1,1]$ as $[-1,0]\cup [0,1]$, and the formulas:
\begin{eqnarray}
\forall x \in [-1,1], \ \exists y \in [-1,0], \ x+y=0 \label{eq:counter31}\\
\forall x \in [-1,1], \ \exists y \in [0,1], \ x+y=0
\label{eq:counter32}
\end{eqnarray}
Note that none of the equations (\ref{eq:counter31}) nor (\ref{eq:counter32}) are true. 
\end{proof}

\subsection{Subdividing either existentially or universally quantified parameter domain}
\label{sec:subdiv1}

In this section, we use Lemma \ref{lem:weaken} to refine the criteria  used in our paving algorithm. Indeed, depending on the direction of the implication,  the computation can be either strengthened or weakened when subdividing, leading to over and under-approximations of the reachable sets. We can then use the fact that for each condition, we have  oracles relying either on over-approximation or under-approximation to make the right choice.
Having smaller sets of values for some parameters means that we are going to have tighter bound estimates for inner and outer approximations. This may be counter-intuitive in the case of inner-approximations, but comes from a tighter evaluation of the Jacobian of function $g$, which is key when using Theorem \ref{thm:robustAE} and its extensions of \cite{hscc24,nahs25}. 
In this section, for the sake of brevity we focus on refining the criteria of Lemma \ref{lem:infirst} and Lemma \ref{lem:outfirst} and omit the  statements corresponding to the other two criteria. But we state the oracles obtained for all four criteria in Table~\ref{tab:summary_refinedoracles}.

Let us first note that cases (1) and (2) of Lemma \ref{lem:weaken}  imply that one can always subdivide the first universally quantified parameters when computing  $\R(\X^\forall,\capitalp,\G)$ (as in Lemma \ref{lem:infirst}) and the first existentially quantified parameters when computing $\R(\X^\exists,\neg \capitalp,\G^\complement)$ (as in Lemma \ref{lem:insecond}). 

For the other parameters, $\capitalp_2$ to $\capitalp_{2l}$, we must be a little more cautious. These are dealt with cases (3) and (4) of Lemma \ref{lem:weaken}. Case (3) of Lemma \ref{lem:weaken} shows that subdividing any or all of the existentially quantified parameters $\capitalp_2,\capitalp_4,\dots,\capitalp_{2l}$, and taking the union of the corresponding sets $\R(\X^\forall, \capitalp^{(i_2,\dots,i_{2l})},\G)$, using notation $\capitalp^{(i_2,\dots,i_{2l})}$ for the $i_2$th subdivision for $\capitalp_2$, $\dots$, $i_{2l}$th subdivision for $\capitalp_{2l}$, will only lead to an inner-approximation of $\R(\X^\forall, \capitalp,\G)$. 
This is spelled out in Lemma \ref{lem:infirstrefinement} in which we refine, by subdivision of parameters, the computation of $\R^-(\X^\forall,\capitalp,\G)$ that appears in Lemma \ref{lem:infirst}:

\begin{lem}
\label{lem:infirstrefinement}
Suppose $\capitalp_j\supseteq \bigcup\limits_{i_j=1}^{k_j} \capitalp_j^{i_j}$ ($j=2,\dots,2l$), then the set $\bigcup_{i_2,\dots,i_{2l}=1}^{k_2,\dots,k_{2l}}
\R^-(\X^\forall,\capitalp^{(i_2,\dots,i_{2l})},\G)$ is inner-approximating $\R^-(\X^\forall,\capitalp,\G)$,  
where $\capitalp^{(i_2,\dots,i_{2l})}$ is the sequence of quantified sets 
$$\left(\capitalp_1^\forall,(\capitalp_2^{i_2})^\exists,\capitalp_3^\forall,\dots,\capitalp_{2l-1}^\forall,(\capitalp_{2l}^{i_{2l}})^\exists\right).$$
\end{lem}

\begin{proof}
By case (3) of Lemma \ref{lem:weaken}, taking $j=2l$ and supposing that $\capitalp_{j}=\bigcup\limits_{i_{2l}=1}^{k_{2l}} \capitalp_j^i$, we see that, for any $\zeta$, free variables of some predicate $\Phi$, and any $i_{2l}=1,\dots,k_{2l}$: 
\begin{multline*}
\lbrace \zeta\in \R^m \mid Q_1 x_1 \in \capitalp_1, \dots, Q_{j-1} x_{j-1} \in \capitalp_{j-1}, \exists x \in \capitalp_j^i, \ \Phi \rbrace \\
\subseteq
\lbrace \zeta \in \R^m \mid Q_1 x_1 \in \capitalp_1, \dots, Q_{j-1} x_{j-1} \in \capitalp_{j-1}, \exists x \in \capitalp_j, \ \Phi\rbrace
\end{multline*}
Applying this with $j=2l$, to $x_1=\x$, $Q_1=\forall$, $x_{i+1}=p_i$ for $i=1,\dots,2l-1$, $x_{2l}=p_{2l}$ and $\Phi=\left(\exists \z \in \G, \ f(\x,\p)-z=\zeta\right)$, we get, for all $i_{2l}=1,\dots,k_{2l}$:
$$
\R^-(\X^\forall,\capitalp^{(i_{2l})},\G)\subseteq \R^-(\X^\forall,\capitalp,\G)
$$
\noindent where $\capitalp^{(i_{2l})}=(\capitalp_1^\forall,\capitalp_2^\exists,\dots,\capitalp_{2l-1}^\forall,{(\capitalp_{2l}^{i_{2l}})}^\exists)$. 

Then by induction on $j=l, l-1,\dots,1$, we successively show that,
$$\R^-(\X^\forall,\capitalp^{(i_{2j},\dots,i_{2l})},\G)
\subseteq \R^-(\X^\forall,\capitalp,\G)$$
\noindent with, 
$$\capitalp^{(i_{2j}, \dots, i_{2l})}=(\capitalp_1^\forall,\capitalp_2^\exists,\dots,\capitalp_{2j-1}^\forall,{(\capitalp_{2j}^{i_{2j}})}^\exists, \dots, \capitalp_{2l-1}^\forall,{(\capitalp_{2l}^{i_{2j}})}^\exists).$$ 

Therefore, 
$$\R^-(\X^\forall,\capitalp^{(i_2,\dots,i_{2l})},\G)
\subseteq \R^-(\X^\forall,\capitalp,\G)$$
for all $i_{2j}=1,\dots,k_{2j}$, $j=1,\dots,l$. 
Taking the union over all $j$ and $i_{2j}$ gives the result. 
\end{proof}

Dually, case (4) of Lemma \ref{lem:weaken} shows  that subdividing the universally quantified sets $\capitalp_1,\capitalp_3,\dots,\capitalp_{2l-1}$ and taking the intersection of the resulting sets  leads to an outer-approximation of $\R(\X^\exists, \capitalp,\G)$. Again, taking the intersection of the outer-approx\-imation of these sets will generally give a better overall outer-approximation. This applies to the oracles in which we need to compute $\R^+(X^\exists,\capitalp,\G)$ (as in Lemma \ref{lem:outfirst}) and $\R^+(\X^\exists,\neg \capitalp,\G^\complement)$ (as in Lemma \ref{lem:insecond}). 
This refinement of the approximation of $\R^+(X^\exists,\capitalp,\G)$  is the object of Lemma \ref{lem:outfirstrefinement}, which result and proof are dual to Lemma \ref{lem:infirstrefinement}.
\begin{lem}
\label{lem:outfirstrefinement}
Suppose $\capitalp_j=\bigcup\limits_{i_j=1}^{k_j} \capitalp_j^{i_j}$ ($j=1,\dots,2l-1$), then the set $\bigcap_{i_1,\dots,i_{2l-1}=1}^{k_1,\dots,k_{2l-1}}
\R^+(\X^\exists,\capitalp^{(i_1,\dots,i_{2l-1})},\G)$ outer-approximates $\R^+(\X^\exists,\capitalp,\G)$ 
where $\capitalp^{(i_1,\dots,i_{2l-1})}$ is the sequence of quantified sets 
$$\left((\capitalp_1^{i_1})^\forall,\capitalp_2^\exists,(\capitalp_3^{i_3})^\forall,\dots,(\capitalp_{2l-1}^{i_{2l-1}})^\forall,\capitalp_{2l}^\exists\right).$$
\end{lem}

\subsection{Combining the subdivision of existentially and universally quantified parameters}
\label{sec:subdiv2}

Ideally, one would like to subdivide all parameters,  universally and existentially quantified, for all criteria at hand. Unfortunately, as previously discussed, subdividing the universally quantified parameters  generally results in (strict) outer-approximations, case (4) of Lemma \ref{lem:weaken}, while subdividing the existentially quantified parameters  generally results in (strict) inner-approximations, case (3) of Lemma \ref{lem:weaken}. Combining both would  result in neither an outer nor an inner approximation, and thus would not be usable. 

However, we can use the idea of subdivision, but  taking midpoints of subdivisions for the values of the existential parameters, resulting in a problem which is purely universally quantified and which can then be  subdivided without over-approximation: this results in an inner-approximation, as shown in Lemma \ref{lem:infirstrefinementbis}. Dually, we can take midpoints of subdivisions of universally quantified parameters and subdivide the existentially quantified ones, this will result in an outer-approximation, as shown in Lemma \ref{lem:outfirstrefinementbis}. 

These results extend a similar idea used in the particular case of one alternation of quantifiers ~\cite{GSJ01}. 

We note $\check{\capitalp}_i$ for any point (generally the middle point) of set $\capitalp_i$, for $i=1,\dots,2l$. Then we can refine  Lemma \ref{lem:infirstrefinement} as follows: 

\begin{lem}
\label{lem:infirstrefinementbis}
Suppose $\capitalp_j= \bigcup\limits_{i_j=1}^{k_j} \capitalp_j^{i_j}$ ($j=1,\dots,2l$), then, we can compute an inner-approximation $\R^-(\X^\forall,\capitalp,\G)$ as: 
$$\bigcup_{i_2,\dots,i_{2l}=1}^{k_2,\dots,k_{2l}} \bigcap_{i_1,\dots,i_{2l-1}=1}^{k_1,\dots,k_{2l-1}}
\R^-(\X^\forall,\capitalp'^{(i_1,i_2,\dots,i_{2l-1},i_{2l})},\G),$$ 
where $\capitalp'^{(i_1,i_2,\dots,i_{2l-1},i_{2l})}$ is the sequence of quantified sets 
$$\left((\capitalp_1^{i_1})^\forall,(\check{\capitalp}_2^{i_2})^\exists,(\capitalp_3^{i_3})^\forall,\dots,(\capitalp_{2l-1}^{i_{2l-1}})^\forall,({\check{\capitalp}_{2l}^{i_{2l}}})^\exists\right).$$
\end{lem}

\begin{proof}
By Lemma \ref{lem:infirstrefinement}, we already know, since $\capitalp_j \supseteq \bigcup_{i_j=1}^{k_j} {\check{\capitalp}_j^{i_j}}$, that 
an inner-approximation $\R^-(\X^\forall,\capitalp,\G)$ can be computed as  
$\bigcup_{i_2,\dots,i_{2l}=1}^{k_2,\dots,k_{2l}}
\R^-(\X^\forall,\capitalp'^{(i_2,\dots,i_{2l})},\G)$, 
where $\capitalp'^{(i_2,\dots,i_{2l})}$ is the sequence of quantified sets 
$$\left(\capitalp_1^\forall,(\check{\capitalp}_2^{i_2})^\exists,\capitalp_3^\forall,\dots,\check{\capitalp}_{2l-1}^\forall,(\check{\capitalp}_{2l}^{i_{2l}})^\exists\right).$$

Now, each quantified problem $\R^-(\X^\forall,\capitalp'^{(i_2,\dots,i_{2l})},\G)$ is a purely universally quantified problem, since the existentially quantified ones can only take one value. 
By case (2) of Lemma \ref{lem:weaken}, we deduce that an inner-approximation $\R^-(\X^\forall,\capitalp,\G)$ can be computed as  
$$\bigcup_{i_2,\dots,i_{2l}=1}^{k_2,\dots,k_{2l}} \bigcap_{i_1,\dots,i_{2l-1}=1}^{k_1,\dots,k_{2l-1}}
\R^-(\X^\forall,\capitalp'^{(i_1,i_2,\dots,i_{2l-1},i_{2l})},\G).$$
\end{proof}

\begin{rem}
The intersection and union above are the set-theoretic intersection and union, and it is simpler, algorithmically, to exercise the test of Lemma \ref{lem:infirst} by checking that $0$ belongs to each $$\bigcap_{i_1,\dots,i_{2l-1}=1}^{k_1,\dots,k_{2l-1}} 
\R^-(\X^\forall,\capitalp'^{(i_1,i_2,\dots,i_{2l-1},i_{2l})},\G).$$
\noindent i.e. that the logical formula  $\bigvee\limits_{i_{2j}} \bigwedge\limits_{i_{2j-1}} \left[\{\0\} \subseteq \Rin((\capitalp_1^{i_1})^\forall,(\check{\capitalp}_2^{i_2})^\exists,\dots)\right]$ is satisfied.  
\end{rem}

And again, we can subdivide both the existentially quantified variables together with the quantified variables as follows, giving a refinement of the criterion of Lemma \ref{lem:outfirstrefinement}: 

\begin{lem}
\label{lem:outfirstrefinementbis}
Suppose $\capitalp_j= \bigcup\limits_{i_j=1}^{k_j} \capitalp_j^{i_j}$ ($j=1,\dots,2l$), then, we can compute an outer-approximation $\R^+(\X^\exists,\capitalp,\G)$ as: 
$$\bigcap_{i_1,\dots,i_{2l-1}=1}^{k_1,\dots,k_{2l-1}}
\bigcup_{i_2,\dots,i_{2l}=1}^{k_2,\dots,k_{2l}} 
\R^+(\X^\exists,\capitalp'^{(i_1,i_2,\dots,i_{2l-1},i_{2l})},\G),$$ 
where $\capitalp'^{(i_1,i_2,\dots,i_{2l-1},i_{2l})}$ is the sequence of quantified sets 
$$\left((\check{\capitalp}_1^{i_1})^\forall,({\capitalp}_2^{i_2})^\exists,(\check{\capitalp}_{i_3}^{i_3})^\forall,\dots,(\check{\capitalp}_{i_{2l-1}}^{i_{2l-1}})^\forall,({{\capitalp}_{2l}^{i_{2l}}})^\exists\right).$$
\end{lem}

\begin{rem}
For exercizing the test of Lemma \ref{lem:outfirst} it is better to use the set theoretic union in the formula of Lemma \ref{lem:outfirstrefinementbis} although, it is indeed correct (as an outer-approximation), to convexify this union. 
\end{rem}

We summarize these ideas in Table~\ref{tab:summary_refinedoracles}.

\begin{table*}
 \caption{Summary of the improvements a la Sections~\ref{sec:subdiv1} and ~\ref{sec:subdiv2} for each oracle.}
    \label{tab:summary_refinedoracles}
    \centering
    \begin{tabular}{ccc}
    \toprule
    Oracle & Subdivision Section~\ref{sec:subdiv1} & Subdivision/points Section~\ref{sec:subdiv2}  \\
    
    \midrule
    $\oraclein(\X)$ (using $\capitalp$, $\G$) & $\bigvee\limits_{i_{2j}} \left[\{\0\} \subseteq \Rin(\X,\capitalp_1^\forall,(\capitalp_2^{i_2})^\exists,\dots,\G)\right]$ & $\bigvee\limits_{i_{2j}} \bigwedge\limits_{i_{2j-1}} \left[\{\0\} \subseteq \Rin(\X,(\capitalp_1^{i_1})^\forall,(\check{\capitalp}_2^{i_2})^\exists,\dots,\G)\right]$ \\
    $\oracleout(\X)$ (using $\capitalp$, $\G$) & $\bigwedge\limits_{i_{2j-1}} \left[\{\0\} \cap \Rout(\X,(\capitalp_1^{i_1})^\forall, \capitalp_2^\exists,\dots,\G) = \emptyset\right]$  & $\bigwedge\limits_{i_{2j-1}} \bigvee\limits_{i_{2j}} \left[\{\0\} \cap \Rout(\X,(\check{\capitalp}_1^{i_1})^\forall, (\capitalp_2^{i_2})^\exists,\dots,\G) = \emptyset\right]$ \\
    $\oraclein(\X)$ (using $\neg\capitalp$, $\G^\complement$) & $\bigwedge\limits_{i_{2j}} \left[\{\0\} \cap \Rout(\X,\capitalp_1^\exists,(\capitalp_2^{i_2})^\forall,\dots,\G^\complement) = \emptyset\right]$ & $\bigwedge\limits_{i_{2j}} \bigvee\limits_{i_{2j-1}} \left[\{\0\} \cap \Rout(\X,(\capitalp_1^{i_1})^\exists,(\check{\capitalp}_2^{i_2})^\forall,\dots,\G^\complement) = \emptyset\right]$ \\
    $\oracleout(\X)$ (using $\neg\capitalp$, $\G^\complement$) & $\bigvee\limits_{i_{2j-1}} \left[\{\0\} \subseteq \Rin(\X,(\capitalp_1^{i_1})^\exists, \capitalp_2^\forall,\dots,\G^\complement)\right]$ & $\bigvee\limits_{i_{2j-1}} \bigwedge\limits_{i_{2j}} \left[\{\0\} \subseteq \Rin(\X,(\check{\capitalp}_1^{i_1})^\exists, (\capitalp_2^{i_2})^\forall,\dots,\G^\complement)\right]$ \\
    \bottomrule
    \end{tabular}
\end{table*}

\begin{rem}
Guaranteed numerical (e.g. set-based) methods are intrinsically limited when the target set has zero measure, while non guaranteed ones generally provide non-solutions. This is compatible with the fact that for QE problem with well-behaved (e.g. continuously differentiable) functions appearing in predicates, paving methods like ours will converge in volume to the solution as the amount of subdivisions increases. For zero measure sets, the method will generally only provide an outer-approximation, with an empty inner-approximation. 
Some methods use inflation or weakenings, as in $\delta$-SMT  \cite{dReal}, that circumvent zero-measure solution sets, by inflating them or relaxing the quantified constraints, which  could be combined with our algorithm.
\end{rem}

\section{Implementation and benchmarks}

The four criteria of Lemma \ref{lem:infirst}, \ref{lem:outfirst}, \ref{lem:insecond} and \ref{lem:outsecond} and the refinement methods of Lemmas \ref{lem:infirstrefinementbis} and \ref{lem:outfirstrefinementbis}  have been implemented in Julia, using packages \texttt{IntervalArithmetic} and \texttt{LazySets}. 
All examples were computed using our prototype\footnote{\url{https://github.com/NuwanHerathM/PaveReach/releases/tag/HSCC26}}, on a 12th Gen Intel(R) Core(TM) i7-12700H processor with a clock rate of 3.50 GHz and a 16 Gio RAM.
In the figures, for a quantified set $\Sigma$, we represent $\Sigma^-$ in green, $\left(\Sigma^\complement\right)^-$ in cyan and $\Sigma^\Delta$ in yellow. So, an inner approximation of $\Sigma$ if given by $\Sigma^-$ in green and an outer approximation of $\Sigma$ is given by $\Sigma^- \cup \Sigma^\Delta$ in green and yellow.

In the sequel, we discuss results, first on the running example in  Section \ref{sec:protorunning}, then on a classical stability problem in control in Section \ref{sec:protostability}, robust reachability in Section \ref{ex:forwarddubins}, and collision problems in Sections \ref{sec:protocircle} and \ref{sec:protorobot}. Note that the prototype is in no way optimized, and runs purely interpreted, and not compiled. 

\subsection{Running example}

\label{sec:protorunning}

We report in Table~\ref{tab:timing_infirst_outfirst}
the evolution of the precision of the paving algorithm, when $\epsilon_\X$, the minimum size of the paving of the domain of $x$, varies. In the first two columns, we report the results for the algorithm without subdivision on the parameters.
As expected, when $\epsilon_\X$ decreases, the percentage of the volume of the  domain $\X$ that is neither declared in $\Sigma^\complement$ nor in $\Sigma$, which is the undecided part of the domain, decreases. But this  converges to about 6\% when $\epsilon_\X$ is of the order of $10^{-4}$. Indeed,  the fact that the parameter $p_1$ is still an interval of width $1/4$ induces an imprecision in the calculations of $\R^-$ and $\R^+$ in the criteria of Lemma \ref{lem:infirst} and \ref{lem:outfirst}, that prevent the algorithm from concluding. 
Figure~\ref{fig:running_example_all_epsilon} also represents the paving of domain $X$ obtained for decreasing values of $\epsilon_\X$. We  also observe that the proportion of undecided domain (in yellow) decreases. Note that we only present the results of the criteria of Lemmas \ref{lem:infirst} and \ref{lem:outfirst}, the others being equivalent in this one-dimensional example. 
\begin{figure}
    \centering
    \includegraphics[width=\linewidth]{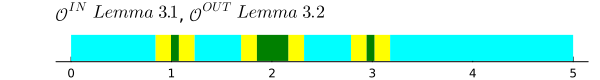}
    \includegraphics[width=\linewidth]{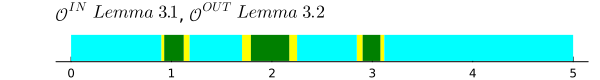}
    \includegraphics[width=\linewidth]{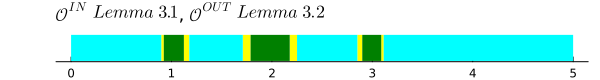}
    \caption{Example~\ref{eq:runningexample}, paving by Algorithm~\ref{alg:paver}  for $\epsilon_\X = 0.1$, $\epsilon_\X = 0.01$ and $\epsilon_\X = 0.001$ (from top to bottom).}
    \label{fig:running_example_all_epsilon}
\end{figure}

In the next columns of Table \ref{tab:timing_infirst_outfirst}, we now consider the paving algorithm   using the methods for subdividing parameters of Section \ref{sec:subdiv1}, where only existential or universal quantified parameters are  subdivided, and of Section \ref{sec:subdiv2}, where all parameters are subdivided, with a minimal subdivision size on parameters $\epsilon_\capitalp$ taken equal to the precision for the paving of $\X$, $\epsilon_\X$. 
We observe an improvement of the precision for smaller subdivision sizes, but converging to about 3\% of domain $\X$ being undecided, for $\epsilon_\X=\epsilon_\capitalp$ of about $10^{-4}$. This is because  subdivisions on parameters for  $\oraclein$ are partly ineffective here: the only parameter, $p_1$, is universally quantified, and is not going to be subdivided using Lemma \ref{lem:infirstrefinement}. Only $\oracleout$ is improved here, since Lemma \ref{lem:outfirstrefinement} allows us to subdivide parameter $p_1$. 

Finally, subdividing all quantified parameters using Lemma \ref{lem:infirstrefinementbis} and \ref{lem:outfirstrefinementbis} allows us to converge quickly to a near exact solution, also with a less significant increase of computation time. This could appear as surprising, but we believe this is due, among other things, to the evaluation of the Jacobian being simpler when half of the variables are fixed values, and to the fact that the subdivision process in $\oracleout$ loses much more time discarding subsets of $\X$ from the solution set, than with the extra refinement method of Section \ref{sec:subdiv2}, because of its more imprecise oracle. 

Tables \ref{tab:running_example_precision_subdivision}
and \ref{tab:running_example_precision_refinement} report the impact of choices of $\epsilon_\X$ and $\epsilon_\capitalp$. Parameter $p_1$ is of order 2, and variable $x$ of order 3, meaning that  with subdivisions the enclosure of the Jacobian should improve linearly for $p_1$, and quadratically for $x$. But the size of intervals in the paving algorithm for $X$  limits the  effect of parameter subdivision. Moreover, in this particular case, parameter subdivision is only effective for $\oracleout$ and not for $\oraclein$. 
\begin{table*}
    \centering
    \caption{Proportion of undecided domain and timings for Example~\ref{eq:runningexample}, with and without parameter subdivision.}
    \label{tab:timing_infirst_outfirst}
    \begin{tabular}{c|cc|cc|cc}
    \toprule
    \multirow{2}{*}{$\epsilon_\X = \epsilon_\capitalp$} & \multicolumn{2}{c} { Alg.~\ref{alg:paver} (Lemmas~\ref{lem:infirst} and~\ref{lem:outfirst})} & \multicolumn{2}{c}{Section~\ref{sec:subdiv1} (Lemmas~\ref{lem:infirstrefinement} and~\ref{lem:outfirstrefinement})} & \multicolumn{2}{c}{Section~\ref{sec:subdiv2} (Lemmas~\ref{lem:infirstrefinementbis} and~\ref{lem:outfirstrefinementbis})} \\
    & Undecided & Time (s) & Undecided  & Time (s) & Undecided & Time (s) \\
    \midrule
    $10^{-1}$ & 19\% & 0.001 & 15.6\% &0.002  & 12.5\% & 0.004  \\
    $10^{-2}$ & 7.0\%& 0.003 & 4.3\% & 0.027  & 1.2\% & 0.064  \\
    $10^{-3}$ & 6.1\% & 0.021 & 3.0\%  & 2.0  & 0.1\%  & 1.0  \\
    $10^{-4}$ & 6.0\% & 0.158 &  3.0\% & 233 & 0.0\% & 31  \\
    \bottomrule
    \end{tabular}
\end{table*}
Figure~\ref{fig:running_example_0.001_refinement} represents  the paving of domain $X$ obtained, from top to bottom, without parameter subdivision, with the approach of Section~\ref{sec:subdiv1} and with the approach of Section~\ref{sec:subdiv2}. 
\begin{figure}
    \centering
    \includegraphics[width=\linewidth]{figures/ex_3-2_runningexample_11_0.001_0.001.png}
    \includegraphics[width=\linewidth]{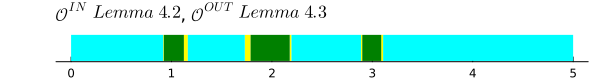}
    \includegraphics[width=\linewidth]{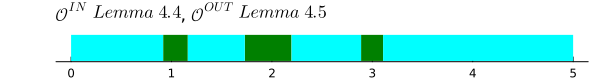}
    \caption{Example~\ref{eq:runningexample}, paving for $\epsilon_\X = \epsilon_\capitalp = 0.001$.}
    \label{fig:running_example_0.001_refinement}
\end{figure}

\begin{table}
    \centering
    \caption{Proportion of the undecided domain for the running example using Lemma~\ref{lem:infirstrefinement} and Lemma~\ref{lem:outfirstrefinement}.}
    \label{tab:running_example_precision_subdivision}
    \begin{tabular}{c|cccc}
    \toprule
    \diagbox{$\epsilon_\X$}{$\epsilon_\capitalp$} & $10^{-1}$ & $10^{-2}$ & $10^{-3}$ & $10^{-4}$ 
    \\
    \hline
    $10^{-1}$ & 15.6\% & 15.6\% & 15.6\% & 15.6\% 
\\
    $10^{-2}$ & 4.3\% & 4.3\% & 4.3\% & 4.3\% 
    \\
    $10^{-3}$ & 3.2\% & 3.0\% & 3.0\% & 3.0\% 
    \\
    $10^{-4}$ & 3.2\% & 3.0\% & 3.0\% & 3.0\% 
    \\
    \bottomrule
    \end{tabular}
\end{table}

\begin{table}
    \centering
    \caption{Proportion of the undecided domain for the running example using Lemma~\ref{lem:infirstrefinementbis} and Lemma~\ref{lem:outfirstrefinementbis}.}
    \label{tab:running_example_precision_refinement}
    \begin{tabular}{c|cccc}
    \toprule
    \diagbox{$\epsilon_\X$}{$\epsilon_\capitalp$} & $10^{-1}$ & $10^{-2}$ & $10^{-3}$ & $10^{-4}$ \\
    \hline
    $10^{-1}$ & 12.5\% & 12.5\% & 12.5\% & 12.5\% \\
    $10^{-2}$ & 2.1\% & 1.2\% & 1.2\% & 1.2\% \\
    $10^{-3}$ & 1.0\% & 0.2\% & 0.1\% & 0.1\% \\
    $10^{-4}$ & 0.9\% & 0.1\% & 0.0\% & 0.0\% \\
    \bottomrule
    \end{tabular}
\end{table}

\subsection{Stability of a controller}

\label{sec:protostability}

In this example from \cite{intervalconstraintsolving} (\emph{GarloffGraf1} extracted from \cite[Example~1]{garloffgraf}), we consider the static output-feedback problem presented in \cite{QEcontrol}. Given a linear ODE, the problem is to find a proportional controller, given by coefficients $v$ and $w$, that stabilizes it. It leads to looking for the parameters $v$ and $w$ which ensure the stability of the polynomial
$p(s) = s^3 + vs^2 + (w - 5v -13) s + w.$

Using the Li\'{e}nard-Chipart criterion, it is transformed into the equivalent system
$
	\{
		v, w > 0, 
		- 5v^2 - 13v + vw - w > 0
	\}.
$
The paving of 
$	\Sigma = \{ v \in [2,10], w \in [40,50] \, | \, \exists z \in [10^{-7}, 100], 
    - 5v^2 - uv + vw - w - z = 0 \},$
obtained in 25 ms using Algorithm \ref{alg:paver} without subdivision and for $\epsilon_\X = 0.1$,  with 1.3\% of the domain being undecided, is presented in Figure~\ref{fig:garloffgraf}.
Note that the more complex quantified problem, obtained for stabilizing an (similar) uncertain ODE,  
$\{ v \in [2,10], w \in [40,50] \, | \, \forall u\in [12.5,13.5], \ \exists z \in [10^{-7}, 100], 
    - 5v^2 - 13v + vw - w - z = 0 \}$, is solved in about the same time, without subdivision either. 
\begin{figure}
    \centering
    \includegraphics[width=\linewidth,height=3.5cm]{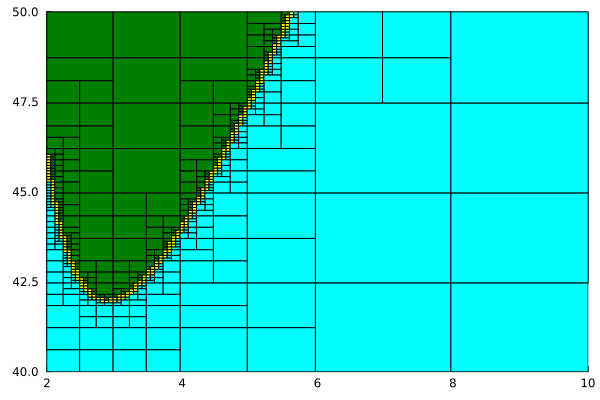}
    \caption{\emph{GarloffGraf1}, paving for $\epsilon_\X = 0.1$.}
    \label{fig:garloffgraf}
\end{figure}

\subsection{Dubins vehicle and generalized reachability}

\label{ex:forwarddubins}

We consider again Example 1 of \cite{hscc24}, used for sample calculation  in Section \ref{ex:Dubbins}. 
We are now interested in approximations of sets  of the form `` $(x,y,\theta)$ such that there exists a control $a$, there exists an initial condition $(x_0,y_0,\theta_0)$, for all perturbations $b_1$, there exists a time $t$ at which $(x,y,\theta)=\phi(t;(x_0,y_0,\theta_0),a,b_1)$''. 
We use  as in \cite{hscc24} an approximate model of the flow of the system, 
and are thus looking at the solution set of the quantified problem: 
\begin{multline*}
\Sigma =
\left\lbrace (x,y,\theta) \in [-2,4] \times [-4,4]\times [-0.5,0.5] \, | \, \right. \\
\exists \epsilon_7 \in [-1,1], \ \exists \epsilon_1 \in [-10,10], \ \exists \epsilon_4 \in [-10,10], \ \exists \epsilon_6 \in [-20,20], \\ \forall \epsilon_2 \in [-1,1],  \ \exists \epsilon_3 \in [-1,1], \ \exists \epsilon_5 \in [-1,1], \ \exists t \in [0,2], \\
x=0.1\epsilon_1+(1+0.01\epsilon_2)t+1.31 \ 10^{-7} \epsilon_3 t^2
\ \wedge 
y=0.1\epsilon_4+(0.01\epsilon_6\\ \left. +0.01\epsilon_7 t) t
+(0.005 \epsilon_5) t^2 \ \wedge 
\theta(t)=0.01\epsilon_6+0.01\epsilon_7 t
\right\rbrace
\end{multline*}
\noindent with much larger initial sets than the similar example  in \cite{hscc24}. 

Then, for $\epsilon_\X = \epsilon_\capitalp = 0.01$ with the subdivision method of Section \ref{sec:subdiv1}, we get $[-0.96, 2.96] \subseteq R_{\exists \forall \exists}(x) \subseteq [-1.01, 3.01]$, computed in 2.6 ms, with only 1.5\% of the domain being undecided ; $[-0.99, 0.99] \subseteq R_{\exists \forall \exists}(y) \subseteq [-1.94, 1.94]$, in 22 ms with 24\% of the domain undecided ; and $[-0.19, 0.19] \subseteq R_{\exists \forall \exists}(\theta) \subseteq [-0.23, 0.23]$, in 1.3 ms, with 7.8\% of the domain being undecided, see Figure~\ref{fig:dubins_individual_components}.

\begin{figure}
    \centering
    \includegraphics[width=\linewidth]{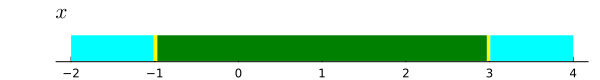}
    \includegraphics[width=\linewidth]{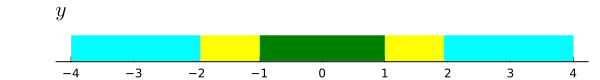}
    \includegraphics[width=\linewidth]{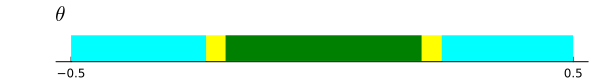}
    \caption{Dubins example, paving of $\Sigma$ with $\epsilon_\X = \epsilon_\capitalp = 0.01$.}
    \label{fig:dubins_individual_components}
\end{figure}

\subsection{Circle collision problem}

\label{sec:protocircle}

We consider the simple problem of collision in the context of camera motion planning from~\cite{intervalconstraintsolving}. 
The problem to solve is, for $r_1 = 2.5$ and $d_1 = 0.5$, to determine $(x,y) \in [-5,5]^2$ such that:
$$
\forall t \in [-\pi,\pi], \ \sqrt{(r_1 \sin t - x)^2 + (r_1 \cos t - y)^2} \geq d_1.
$$

The computation of the paving presented in Figure~\ref{fig:circle_collision} took 25s with only 6.0\% of the domain being undecided, for $\epsilon_\X = \epsilon_\capitalp = 0.1$,  using the refinement algorithm of Section \ref{sec:subdiv2}. 
For $\epsilon_\X = \epsilon_\capitalp = 0.01$, we reach 0.4\% of the domain being undecided in $3100$s.
\begin{figure}
    \centering
    \includegraphics[width=\linewidth,height=3.5cm]{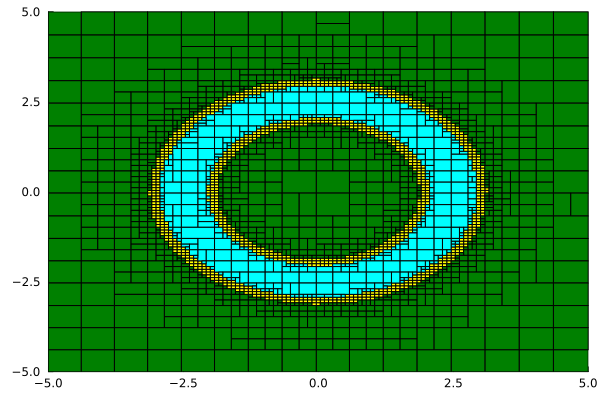}
    \caption{Circle collision problem, paving  with $\epsilon_\x = \epsilon_\capitalp = 0.1$, using Lemma~\ref{lem:infirstrefinementbis} and Lemma~\ref{lem:outfirstrefinementbis}.}
    \label{fig:circle_collision}
\end{figure}

\subsection{Robot collision problem}

\label{sec:protorobot}

In this problem extracted from \cite{intervalconstraintsolving}, a mobile robot arm is composed of three
segments of respective lengths $d_1$, $d_2$ and $d_3$. 
The trajectory of the robot’s hand $P(t)$ is determined by three
angle functions $\alpha_1(t)$, $\alpha_2(t)$ and $\alpha_3(t)$. The problem is to find all points in the
2D space that do not collide with the robot’s hand, i.e. computing all the $(x,y)$ in $[0,5]$
 such that: 
$$
\forall t\in [0,2], \ \sqrt{(x-P_x(t))^2+(y-P_y(t))^2}\geq d
$$
\noindent with $\alpha_1(t) =  t+\pi/4 $, $\alpha_2(t) =  2t-1$, $\alpha_3(t) = 0.2t+0.1$ and 
$$
\begin{array}{rcl}
P_x(t) & = & d_1sin \alpha_1(t)+d_2sin(\alpha_1(t)+\alpha_2(t)-\pi) \\
            &&+d_3sin(\alpha_1(t)+\alpha_2(t)+\alpha_3(t)) \\
P_y(t) & = & d_1cos\alpha_1(t)+d_2cos(\alpha_1(t)+\alpha_2(t)-\pi) \\
            &&+d_3cos(\alpha_1(t)+\alpha_2(t)+\alpha_3(t)) 
\end{array}
$$
\noindent for $d=0.5$, $d_1=1$, $d_2=2$, $d_3=1$.

Using the refinement of Section \ref{sec:subdiv2}, we compute a paving of the solution set in 5.3 s with 9.2\% of the  domain being undecided, for $\epsilon_\X = \epsilon_\capitalp = 0.1$, and in 291 s with 0.8\% of the domain being undecided for $\epsilon_\X = \epsilon_\capitalp = 0.01$. 
Algorithm \ref{alg:paver} (i.e. with no parameter subdivision), results in 93\% of the domain being undecided, computed in 165 ms; using the subdivisions of Section~\ref{sec:subdiv1}, we end up with 52\% of the domain being undecided, in 3.1 s.
Here, the problem has no existentially quantified parameter, so that only the refinement  of Section \ref{sec:subdiv2} is able to characterize a non empty approximation $\Sigma^-$, see Figure~\ref{fig:robot_collision_0.1} for a visual comparison of the  refinements. 
\begin{figure}
    \centering
    \includegraphics[width=0.33\linewidth]{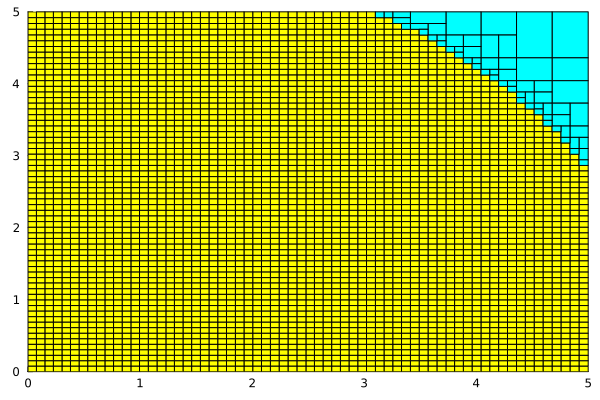}%
    \includegraphics[width=0.33\linewidth]{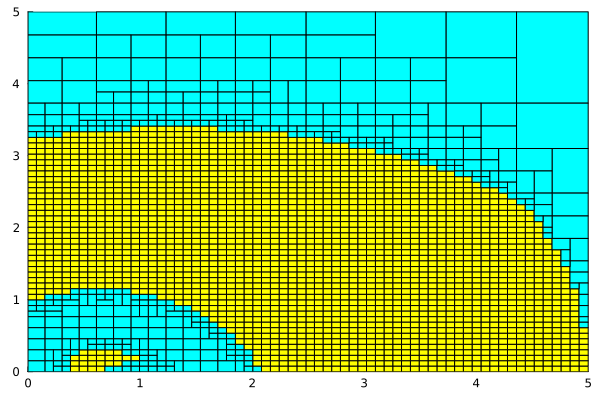}%
    \includegraphics[width=0.33\linewidth]{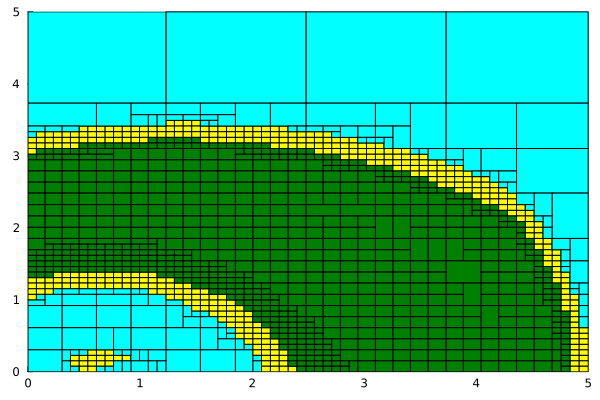}
    \caption{Robot collision problem, $\epsilon_\X = \epsilon_\capitalp = 0.1$, using from left to right, Lemmas~\ref{lem:infirst} and~\ref{lem:outfirst}, \ref{lem:infirstrefinement} and~\ref{lem:outfirstrefinement}, and \ref{lem:infirstrefinementbis} and~\ref{lem:outfirstrefinementbis}.}
    \label{fig:robot_collision_0.1}
\end{figure}

\section{Conclusion}

We proposed a computation of inner and outer approximations of global solutions to general quantified constraint satisfaction problems, and discussed some of the many choices that impact the performance and the quality of the solution set obtained. 
Future work includes extension to fully general quantifier elimination problems, handling also disjunctions of predicates, dealing with more applications, such as controller synthesis, and improving the parameter subdivision strategies. 
We also plan to experiment this general framework using higher-order methods instead of interval methods, to increase both precision and performance. 

\section*{Acknowledgsments}
This work was partially supported by
the SAIF project, funded by the “France 2030” government
investment plan managed by the French National Research
Agency, under the reference ANR-23-PEIA-0006.

\bibliographystyle{ACM-Reference-Format}
{\footnotesize \bibliography{biblio}}

\end{document}